\documentclass[runningheads]{llncs}
\pdfoutput=1

\usepackage{lmodern}

\usepackage[T1]{fontenc}
\usepackage{lineno}

\usepackage{color}
\usepackage{amsmath,amssymb}
\usepackage{thm-restate}
\usepackage{enumitem}
\usepackage{csquotes}
\usepackage{xspace}
\usepackage[hidelinks]{hyperref}
\usepackage[noabbrev,capitalise]{cleveref}
\usepackage{graphicx}
\usepackage{todonotes}

\spnewtheorem{observation}{Observation}{\bfseries}{\itshape}
\newcommand{\lightdfn}[1]{\emph{#1}}  % non-rigorous or prose def.

\newcounter{oldfootnotevalue}
\newcommand{\TODO}[2][\empty]{%
  \unskip%  % Removes one preceeding glue character; typically a space.
  {%
    \setcounter{oldfootnotevalue}{\value{footnote}}%
    \renewcommand{\thefootnote}{{\normalfont\bfseries\sffamily\color{red} TODO}}%
    \footnote{\normalfont\color{red}\mbox{}\hspace{1.2em}\mbox{}\ifx #1\empty\empty\else #1:\ \fi{#2}}%
    \setcounter{footnote}{\value{oldfootnotevalue}}%
  }%
}
\newcommand{\QUESTION}[2][\empty]{%
  \unskip%  % Removes one preceeding glue character; typically a space.
  {%
    \setcounter{oldfootnotevalue}{\value{footnote}}%
    \renewcommand{\thefootnote}{{\normalfont\bfseries\sffamily Q}}%
    \footnote{\normalfont\color[rgb]{0,0,0.75}\ifx #1\empty\empty\else#1:\ \fi{#2}}%
    \setcounter{footnote}{\value{oldfootnotevalue}}%
  }%
}

\newcommand\blfootnote[1]{%
	\begingroup
	\renewcommand\thefootnote{}\footnote{#1}%
	\addtocounter{footnote}{-1}%
	\endgroup
}

\newcommand{\np}{\textsf{NP}\xspace}
\renewcommand{\qed}{\hfill $\square$}
\newcommand{\subcaptionheading}[1]{\textbf{\textsf{#1}}\ }

\newenvironment{conditions}%
               {\begin{enumerate}[label=(C\arabic*), left=0pt]}%
               {\end{enumerate}}
\Crefname{condition}{}{Conditions}
\creflabelformat{condition}{#2#1#3}
\crefrangelabelformat{condition}{#3#1#4--#5#2#6}

\newlist{proofcases}{enumerate}{1}%
\setlist[proofcases]{label=(\arabic*),wide=0pt}
\newcommand{\caseheading}[1]{{\itshape\sffamily #1}}
\Crefname{proofcasesi}{Case}{Cases}
               {\begin{enumerate}[label=(\alph*), left=0pt]}%
               {\end{enumerate}}
\Crefname{characterization}{}{Characterizations}
\creflabelformat{characterization}{(#2#1#3)}
\crefrangelabelformat{characterization}{(#3#1#4--#5#2#6)}

\newenvironment{properties}%
               {\begin{enumerate}[label=(P\arabic*), left=0pt,nosep]}%
               {\end{enumerate}}
\Crefname{property}{}{Properties}
\creflabelformat{property}{#2#1#3}
\crefrangelabelformat{property}{#3#1#4--#5#2#6}

\newlist{arcproperties}{enumerate}{1}%
\setlist[arcproperties]{label=({A}\arabic*),left=0pt,nosep}
\Crefname{arcpropertiesi}{}{Properties}
\creflabelformat{arcpropertiesi}{#2#1#3}
\crefrangelabelformat{arcpropertiesi}{#3#1#4--#5#2#6}

\newlist{rainbowconfigurations}{enumerate}{1}%
\setlist[rainbowconfigurations]{label=$\mathcal{C}_{\arabic*}$:, ref=$\mathcal{C}_{\arabic*}$, left=0pt}
\crefname{rainbowconfigurationsi}{}{Configurations}
\creflabelformat{rainbowconfigurationsi}{#2#1#3}
\crefrangelabelformat{rainbowconfigurationsi}{#3#1#4 -- #5#2#6}

\definecolor{red}{rgb}{0.890196,0.101961,0.109804}  % {227,26,28}
\definecolor{blue}{rgb}{0.121569,0.470588,0.705882}  % {31,120,180}
\definecolor{green}{rgb}{0.2,0.627451,0.172549}  % {51,160,44}

\newcommand{\colonename}{red\xspace}

\newcommand{\coltwoname}{blue\xspace}

\title{On plane cycles in geometric multipartite graphs}
\author{Marco Ricci\inst{1}\orcidID{0000-0002-4502-8571} \and
	Jonathan Rollin\inst{1}\orcidID{0000-0002-6769-7098} \and
	André Schulz\inst{1}\orcidID{0000-0002-2134-4852} \and
	Alexandra Weinberger\inst{1}\orcidID{0000-0001-8553-6661}
}
\authorrunning{M. Ricci, et al.}
\institute{FernUniversität in Hagen, Germany
	\email{firstname.lastname@fernuni-hagen.de}}

\graphicspath{{figs/}}

\begin{document}

\maketitle
%\linenumbers
\begin{abstract}
	A geometric graph is a drawing of a graph in the plane where the vertices are drawn as points in general position and the edges as straight-line segments connecting their endpoints. It is plane if it contains no crossing edges.
	We study plane cycles in geometric complete multipartite graphs.
	We prove that if a geometric complete multipartite graph contains a plane cycle of length~$t$, with~$t\geq 6$, it also contains a smaller plane cycle of length at least~$\lceil t/2\rceil +1$.
	We further give a characterization of geometric complete multipartite graphs that contain plane cycles with a color class appearing at least twice.
	For geometric drawings of $K_{n,n}$, we give a sufficient condition under which they have, for each $s\leq n$, a plane cycle of length $2s$.
	We also provide an algorithm to decide whether a given geometric drawing of $K_{n,n}$ contains a plane Hamiltonian cycle in time $O(n \log n + nk^2) + O(k^{5k})$, where $k$ is the number of vertices inside the convex hull of all vertices.	
	Finally, we prove that it is NP-complete to decide if a subset of edges of a geometric complete bipartite graph $H$ is contained in a plane Hamiltonian cycle in~$H$.%
	\blfootnote{\textit{An extended abstract of this work is available at the 51st International Workshop on Graph-Theoretic Concepts in Computer Science (WG2025).}}
%	\keywords{bipartite graphs, geometric graphs, plane cycles, fixed-parameter tractable}
\end{abstract}

\section{Introduction}

A \lightdfn{geometric graph} is a drawing of a graph $G$ in the plane such that its vertices in $V(G)$ are drawn as disjoint points in general position (i.e., no three points lie on a line) and its edges in $E(G)$ are drawn as straight-line segments connecting their respective endpoints.
We associate the vertices of the graph with their corresponding points in the plane and treat them interchangeable.
A geometric graph is called \lightdfn{plane} if it contains no crossing edges. We refer to a geometric graph whose abstract graph
is $G$ as a \lightdfn{drawing of $G$}.
A geometric (host) graph $H$ \lightdfn{contains} a plane graph $G$ if there is a subdrawing of $H$ without edge crossings, whose (abstract) graph is isomorphic to $G$.

Past research has mostly focused on questions where the host graph 
is a geometric complete graph. 
Cabello~\cite{cabello} showed that for a given graph $G$ it is \np-hard to decide whether a geometric complete graph contains a plane $G$.
However, if $G$ is outerplanar, this can be decided in polynomial time; see for example Bose~\cite{bose}.
It is easy to see that each geometric complete graph $H$ contains plane paths and cycles with up to $\lvert V(H)\rvert$ vertices.
This is no longer true for other geometric host graphs.
\Cref{fig:introExamples} shows geometric complete multipartite graphs that do not contain plane cycles of certain lengths.

In this paper we consider geometric complete multipartite graphs as host graphs, with an emphasis on the bipartite case. 
%Unless otherwise mentioned all host graphs are complete multipartite geometric graphs from now on.\todo{AW: As far as I see, we actually always say complete multipartite. We can take this sentence back in and then make use of it, though}
We study whether these host graphs contain a plane cycle of a given length.
We will usually refer to the vertex partition by its associated vertex coloring, where the color classes correspond to the partition classes.
%Using this viewpoint we refer to edges in $H$ as \lightdfn{bichromatic} edges, and to edges not
%in $H$ as \lightdfn{monochromatic} edges.
%vertices in the same partition class have the same color and vertices from different classes have different colors.
Throughout the paper we depict vertices from the same partition class by a common symbol/color in our figures. Moreover, we omit the edges of the host graph from our figures for a better readability, since they can be easily recovered from the vertex symbols and locations.

\begin{figure}
	\centering
	\includegraphics{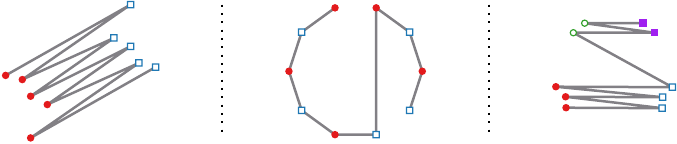}
	\caption{Plane paths in geometric complete multipartite graphs (additional edges of the host graphs are omitted).		
	Each graph contains a plane path on 10 vertices, but no plane 10-cycle.
		In fact, no plane cycle is contained in the graph on the left,  the graph in the middle contains plane cycles of each even length at most $8$, and the graph on the right contains exactly those plane cycles that visit at most one vertex from each color (called rainbow cycles).}
	\label{fig:introExamples}
\end{figure}

There has been much interest in related questions over the years, and collections of many of the obtained results are given in the surveys by Kaneko and Kano~\cite{KK03} and by Kano and Urrutia~\cite{KU21}.
Particular attention has been paid to long plane paths and cycles.
For geometric complete bipartite host graphs, Kaneko and Kano~\cite{KK01} showed that if one color class is sufficiently larger than the other, then any drawing contains a plane path containing all vertices of the smaller color class.
Many known results focus on the special case of geometric complete bipartite host graphs with $n$ vertices of each color where all vertices are in convex position.
Mulzer and Valtr~\cite{MP20} showed that each such host graph contains a plane path of length at least $n+\epsilon n$, for some $\epsilon>0$, improving previous results (cf.~\cite{KPT08,M11}). 
For an upper bound in this case, constructions of such drawings where each plane path has length at most $\frac{4}{3}n+O(\sqrt{n})$ were given by Abellanas et al.~\cite{AGHT03} and Kyn{\v{c}}l, Pach, and T{\'o}th~\cite{KPT08}.
Other special (with respect to vertex placement) geometric graphs were investigated with regard to plane Hamiltonian paths and cycles by Cibulka et al.~\cite{CKMSV09} and Soukup~\cite{Soukup24}. 
For geometric complete multipartite host graphs with an arbitrary number of color classes, Merino, Salazar and Urrutia~\cite{MSU06} gave a lower bound on the length of the longest plane path contained in any such geometric graph and showed it is tight for odd numbers of colors.
Weakening the noncrossing requirement, 1-plane Hamiltonian paths and cycles (i.e.\ paths and cycles of length $2n$) have been studied by Claverol et al.~\cite{COGST18}, and Hamiltonian paths and cycles with the minimal number of crossings have been studied by Kaneko et al.~\cite{KKY00}.

The algorithmic complexity of finding plane paths or cycles is open.
However, Akitaya and Urrutia~\cite{AU90} and Abellanas et al.~\cite{AGHT03} described algorithms that decide whether a geometric complete bipartite host graph with all vertices in convex position has a plane Hamiltonian path in $O(n^2)$ time.
Bandyapadhyay et al.~\cite{DBLP:conf/caldam/BandyapadhyayBB20,BBBN21} gave linear time algorithms to find plane Hamiltonian paths in drawings of complete bipartite graphs, where the vertices are mapped to the real line and the edges are drawn as circular arcs above and below the line.

We point out that a geometric drawing is uniquely determined by the vertex locations in the plane.
%Hence, the answer to the question whether a given complete multipartite geometric host graph $H$ contains a plane $G$ solely depends on $G$, the partition classes of $V(H)$, and the locations of the vertices from $H$ in the plane.
Thus, the question whether a given geometric complete multipartite host graph~$H$ contains a plane $G$ is equivalent to the question of whether the vertices of $G$ can be mapped to colored points in the plane, given by $V(H)$, such that the endpoints of each edge in $G$ are mapped to points of different colors and the resulting %(geometric) 
drawing of $G$ is plane.
%Thus the question is equivalent to the question whether the vertices of $G$ can be mapped to the (colored) points in the plane determined by $V(H)$ such that the endpoints of each edge in $G$ are mapped to points of different colors and the resulting (geometric) drawing of $G$ is plane. 
Both perspectives are used in previous work.
From another perspective, plane cycles can also be interpreted as polygons and there is a significant amount of research on so-called polygonizations (cf.~\cite{DBLP:conf/compgeom/AkitayaKRST19}).

In this paper we study the existence and non-existence of plane cycles of specific lengths in geometric complete multipartite graphs.
In \cref{sec:monotonicity} we consider two extremal classes of such drawings: drawings without any plane cycles and drawings with plane Hamiltonian cycles.
We also investigate a possible \enquote{monotonicity} property of plane cycles, as we think that a  natural candidate for a drawing with a plane cycle of a specific length is a drawing containing a longer plane cycle.
We prove, among others, that if a drawing of a complete multipartite graph contains a plane cycle of length $t$, with $t\geq 6$, then it also contains a shorter plane cycle of length at least $\lceil\frac{t}{2}\rceil + 1$.
In particular, it contains a plane cycle of length four or five.
In \cref{sec:FPT} we give an FPT algorithm to decide whether a geometric bipartite graph $H$ contains a plane Hamiltonian cycle, where the parameter is the number of vertices of $H$ in the interior of the convex hull of all vertices.
%This is a natural generalization of the existence of polynomial time algorithms~\cite{AGHT03,AU90} when the vertices are in convex position.
In \cref{sec:NP} we consider the complexity of a restricted version of the same problem, where some given edges of the bipartite host graph have to be included in the plane Hamiltonian cycle.
We show that the problem is \np-complete. Note that this problem is known to be \np-complete for (uncolored) complete graphs due to Akiyama et al.~\cite{DBLP:conf/compgeom/AkitayaKRST19} and Jiang, Jiang, and Jiang~\cite{DBLP:journals/corr/abs-2108-12812}.

\section{Monotonicity results for plane cycles}\label{sec:monotonicity}

%We call an alternating plane cycle with all points of different colors \lightdfn{rainbow}, and \lightdfn{non-rainbow} if at least two points have the same color.

In this section, we are interested in the following \enquote{monotonicity} question:
Does the existence of a plane cycle of length~$t$ in a geometric complete multipartite graph $H$ imply the existence of a plane cycle of length~$t'$ in $H$ for each $t'$ with $3\le t' \leq t$?
If the given plane cycle is one where all vertices have different colors (called a \lightdfn{rainbow cycle}), then this can be trivially answered with yes.
The question becomes interesting if at least two vertices of the given cycle have the same color; we call a cycle with this property \lightdfn{non-rainbow}. 
As complete bipartite  graphs do not contain rainbow cycles, this is the only sensible case for the case of two colors (for which the question should only be asked for all $t'$ even).
There are geometric complete multipartite graphs that do not contain any plane non-rainbow cycle, including bipartite ones without any plane cycles; %-- including drawings of complete bipartite graphs which do not contain any plane cycle; 
see \cref{fig:introExamples} (left and right) for examples.
Our first result characterizes those drawings.
The characterization is based on four forbidden configurations~$\mathcal{C}_1$~--~$\mathcal{C}_4$, which are depicted in \cref{fig:forbCyclePatterns}.
For a geometric complete multipartite graph $H$ let $\phi:V(H)\to\mathbb{N}$ denote the associated vertex coloring.

\begin{rainbowconfigurations}
	\item Four distinct vertices $u$, $u'$, $v$, $v'\in V(H)$ and two distinct colors $\phi_1$ and $\phi_2$, with $\phi(u)=\phi(u')=\phi_1$, $\phi(v)=\phi(v')=\phi_2$, such that $v$ and $v'$ lie on different sides of the straight line through $u$ and $u'$. \label{rainbow:bicolored}
	\item Four distinct vertices $u$, $u'$, $v$, $v'\in V(H)$ and three distinct colors $\phi_1$, $\phi_2$, and $\phi_3$, with $\phi(u)=\phi(u')=\phi_1$, $\phi(v)=\phi_2$, and $\phi(v')=\phi_3$, such that $v$ and $v'$ lie on different sides of the straight line through $u$ and $u'$. \label{rainbow:pair-line-through}
	\item Four distinct vertices $u$, $u'$, $v$, $v'\in V(H)$  and three distinct colors $\phi_1$, $\phi_2$, and $\phi_3$, with $\phi(u)=\phi_1$, $\phi(u')=\phi_2$, $\phi(v)=\phi(v')=\phi_3$, such that $v$ and $v'$ lie on different sides of the straight line through $u$ and $u'$. \label{rainbow:pair-separated}
	\item Five distinct vertices $u$, $u'$, $v$, $v'$, $w\in V(H)$  and four distinct colors $\phi_1,\ldots,\phi_4$, with $\phi(u)=\phi(u')=\phi_1$, $\phi(v)=\phi_2$, $\phi(v')=\phi_3$ and $\phi(w)=\phi_4$, such that $u$, $u'$, $v$, and $v'$ form a convex quadrilateral with $w$ in its interior. \label{rainbow:five-point}
\end{rainbowconfigurations}

We say that a geometric graph $H$ \emph{contains} configuration $\mathcal{C}_i$, if there are four or five vertices in $H$ satisfying the constraints described by $\mathcal{C}_i$.

\begin{figure}[tb]
	\centering
	\includegraphics[page=1]{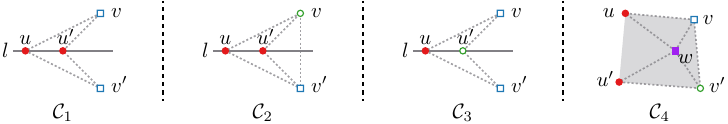}
        \caption{The configurations $\mathcal{C}_1$ -- $\mathcal{C}_4$ which guarantee the existence of non-rainbow plane cycles. $\mathcal{C}_1$ -- $\mathcal{C}_3$: Vertices $v$ and $v'$ lie on different sides of the straight line $l$ through $u$ and $u'$. $\mathcal{C}_4$: Vertices $\{u,u',v,v'\}$ form a convex quadrilateral with $w$ in its interior.}
	\label{fig:forbCyclePatterns}
\end{figure}

\begin{restatable}{theorem}{CharacterizationRainbow}\label{thm:characterization-rainbowcycles}
	A geometric complete multipartite graph $H$ contains a non-rainbow plane cycle if and only if it contains one of the configurations $\mathcal{C}_1$ -- $\mathcal{C}_4$.
	
	There is an algorithm that checks whether $H$ has a non-rainbow plane cycle in time $O(\lvert V(H)\rvert^5)$.
\end{restatable}
\begin{proof}
	A geometric complete multipartite graph which contains one of the configurations \cref{rainbow:bicolored,rainbow:pair-line-through,rainbow:pair-separated,rainbow:five-point} has a non-rainbow plane cycle, as each of configurations admits such a cycle.
	
	So consider a geometric complete multipartite graph $H$ which contains none of the configurations \cref{rainbow:bicolored,rainbow:pair-line-through,rainbow:pair-separated,rainbow:five-point}.
	For the sake of a contradiction suppose that $H$ has a non-rainbow plane cycle~$C$, that is, $C$ contains two vertices of the same color~$\phi_1$.
	The absence of configurations \cref{rainbow:bicolored} and \cref{rainbow:pair-line-through} implies that within the convex hull of the vertices of color $\phi_1$, there is no vertex of a different color.
	This shows that there is a line $\ell$ separating all vertices of color $\phi_1$ from all other vertices of $V(H)$.
	Without loss of generality assume that $\ell$ is horizontal; see \cref{fig:planeCycleMulticolor}.
	Let $uv$ denote the edge in $C$ whose intersection with $\ell$ is to the left of all other intersections between $C$ and $\ell$, with $u$ having color $\phi_1$.
	Such an edge exists as the endpoints of each edge in $C$ have distinct colors.
	Further let $u'$ be the vertex of color~$\phi_1$ in~$C$ that is reached first when traversing $C$ from $u$ in direction of $v$.
	Then $u\neq u'$, since $C$ contains at least two vertices of color $\phi_1$.
	Let $C_1$ denote the corresponding $u$-$u'$-path in $C$ (containing $v$) and let $C_2$ denote the other $u$-$u'$ path in $C$.
	Then the only vertices of color $\phi_1$ in $C_1$ are $u$ and $u'$, and $C_1$ crosses $\ell$ only with the edges incident to these vertices, while $C_2$ might contain more vertices of color $\phi_1$ and more crossings with $\ell$.
	
	Consider the neighbor $w$ of $u$ in $C_2$.
	Then $w\neq u'$ and $w$ is not of color $\phi_1$.
	As the intersection of $uv$ is to the left all other intersections between $C$ and $\ell$ and because $C$ is plane, the vertex $w$ is in the interior of the convex hull of the vertices in $C_1$.
	So there is a set $X$ of one vertex or two adjacent vertices from $C_1$ such that $w$ is in the interior of the convex hull spanned by $\{u,u'\}\cup X$.
	We distinguish few cases to see that we find one of the configurations \cref{rainbow:bicolored,rainbow:pair-line-through,rainbow:pair-separated,rainbow:five-point} among them.	
	If $w$ and some vertex in $X$ have the same color, then we find \cref{rainbow:bicolored} or \cref{rainbow:pair-line-through}.
	Otherwise, we find \cref{rainbow:five-point} (if $|X|=2$ and the vertices in $X$ have different colors), or \cref{rainbow:pair-separated} (if the vertices in $X$ have the same color or $|X|=1$).
	This is a contradiction and, hence, $H$ does not have any non-rainbow plane cycle.
	
	\begin{figure}
		\centering
		\includegraphics{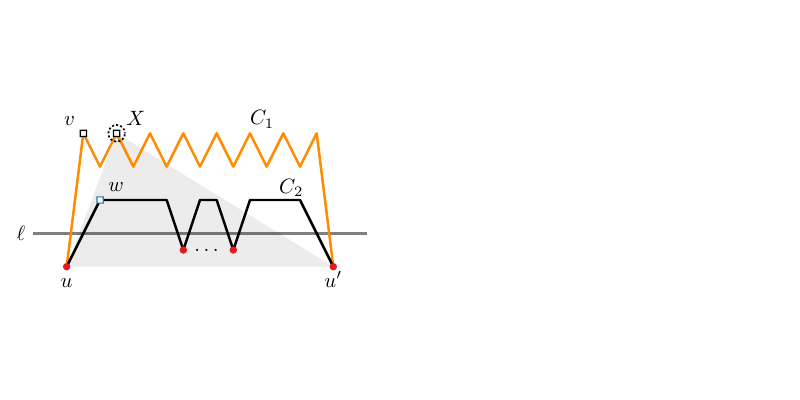}
		\caption{With $X = \{x\}$, the set $\{u, u', w, x\}$ admits \cref{rainbow:bicolored} or \cref{rainbow:pair-separated}, depending on the colors of $x$ and~$w$.}
		\label{fig:planeCycleMulticolor}
	\end{figure}
	
	As each forbidden configuration has only four or five vertices, this also yields an algorithm to check for the existence of non-rainbow plane cycles with running time in~$O(n^5)$.\qed
\end{proof}
	
%\begin{proof}[sketch, full proof in \cref{app:ssec:rainbowcycles}]
%	
%	We now sketch the other direction of the proof in the bipartite case.
%	Consider a geometric complete bipartite graph $H$ which does not contain the configuration~\cref{rainbow:bicolored}.
%	Within the convex hull of the vertices of color $\phi_1$, there is no vertex of the other color, due to the absence of configuration \cref{rainbow:bicolored}.
%	This shows that there is a line $\ell$ separating all vertices of color $\phi_1$ from all vertices of color $\phi_2$.
%	Without loss of generality assume that~$\ell$ is horizontal.
%	Note that each edge of $H$ is crossing $\ell$.
%	Consider any cycle $C$ in~$H$ and let $uv$ denote the edge in $C$ whose intersection with $\ell$ is to the left of all other intersections between $C$ and $\ell$.
%	Let $uw$ and $vw'$ denote the edges in $C$ incident to $uv$.
%	Then $w\neq w'$ as $H$ is bipartite.
%	A brief case analysis shows that either $u$, $v$, $w$, and $w'$ form configuration \cref{rainbow:bicolored} (a contradiction) or $uw$ and $vw'$ cross each other.
%	Hence, $C$ is not plane and $H$ does not contain any plane cycles.
%	\qed
%\end{proof}

As \cref{rainbow:bicolored} is the only forbidden configuration with only two colors, it follows that  any geometric complete bipartite graph $H$ contains a plane cycle if and only if it contains a plane cycle of length $4$.
With the help of \cref{thm:characterization-rainbowcycles} we can further easily check that the geometric graphs in \cref{fig:introExamples} (left and right) do not have non-rainbow plane cycles.
Another consequence of \cref{thm:characterization-rainbowcycles} is that any geometric complete multipartite graph that has some non-rainbow plane cycle also has such a cycle of length~4 or~5.
\Cref{thm:half} below gives a partial answer to the monotonicity question.

To prove this theorem, we first introduce some notation.
It is well-known that every plane geometric cycle $C$ can be augmented to a triangulation $T$; that is a plane graph
whose interior faces are triangles. In this case we say $T$ is a triangulation of $C$.
We call an interior edge $ab$ in a triangulation \lightdfn{flippable}, if its incident triangles $abc$ and $abd$ form a convex quadrilateral.
A \lightdfn{flip} of a flippable edge $ab$, with incident triangles $abc$ and $abd$, replaces $ab$ by $cd$.
This leads to a different triangulation.
Note that any two triangulations on the same point set can be transformed into each other by a sequence of flips~\cite{flipGraphConnected}.

\begin{restatable}{theorem}{ThmHalf}\label{thm:half}
	Let $H$ be a geometric complete multipartite graph which contains a non-rainbow plane cycle $C$ of length $t$, where $t \geq 6$.
	Then $H$ contains a shorter non-rainbow plane cycle of length at least $\lceil\frac{t}{2}\rceil+1$.
\end{restatable}

\begin{proof}
		Without loss of generality we assume that $V(H)=V(C)$. 
	A \lightdfn{principal point} of $C$ is a vertex $p_1$ in $C$ such that for its neighbors $p_2$ and $p_3$ in $C$ the straight-line segment between $p_2$ and $p_3$ does not cross $C$.
	If $p_2p_3$ is an edge in $H$, removing $p_1$ and adding $p_2p_3$ yields a plane cycle $C_p$ of length $t-1\geq \lceil\frac{t}{2}\rceil+1$ (as $t\geq 6$). 
	If $C_p$ is non-rainbow, we call $p$ \lightdfn{good} and we are done.
	We call $p$ \lightdfn{bad} if it is not good, which means either $p_2p_3$ is not an edge in $H$ or it is in $H$ and the corresponding shorter cycle $C_p$ is rainbow.  
	So for the rest of the proof we assume that each principal point is bad.
	We distinguish two cases.
	\begin{proofcases}
		\item\label[proofcase]{case-2:bichromatic-shortcut}
		\caseheading{There exist an edge $uv$ in $H$ but not in $C$ that does not cross any edge of $C$. ($H$ contains a noncrossing chord of $C$.)}
		
		%As $uv$ is in $H$, the vertices $u$ and $v$ have different colors.
		Let  $\pi_1$ and $\pi_2$ denote the two paths in $C$ that connect $u$ and $v$.
		Let $n_1$ be the number of edges in $\pi_1$ and let $n_2$ be the number of edges in $\pi_2$.
		Then $n_1 + n_2 = t$ holds.
		Without loss of generality, assume that $n_1 \geq n_2$.
		Thus, the cycle $C'$ formed by $\pi_1$ together with $uv$ is plane and contained in $H$.
		Moreover, its length is less than $t$ and at least $\lceil\frac{t}{2}\rceil+1$. It remains to prove that $C'$ is non-rainbow.
		
		First assume that $uv$ is in the interior of the region bounded by $C$.
		%We claim that $C'$ is non-rainbow.
		For the sake of a contradiction assume that $C'$ is rainbow.
		Then all vertices in $\pi_1$ are of different colors.
		Since $uv$ is in the interior of $C$, there is at least one principal point $u'$ of $C$ which is in $C'$.
		%		Hence, the neighbors of $u'$ are of different colors.
		%		\TODO{Is this conclusion really that immediate? If there is a principal point in $C$ that is in $\pi_1$ and not $\pi_2$ this is clear from $C'$ being rainbow. But I don't see why there must be a principal point in $C$ that is not in $\pi_2$ (even if it $\pi_1$ is larger and the edge $uv$ does not cross $C$). Or if we then want take a principal point not in $\pi_1$, I don't see why $\pi_1$ shouldn't be an (even) bicolored alternating path. And if we then want to take $u$ and $v$, we first need to argue that they are principal points (that their neighbors are of different colors is only clear to me if $\pi_2$ is alternating).}
		As the neighbors of $u'$ are of different colors, the edge between them is in $H$ and thus the corresponding shorter cycle $C_{u'}$ is in $H$.
		Because $C$ is non-rainbow but $C_{u'}$ is rainbow (as $u'$ is bad by assumption), we conclude that $C$ contains exactly two vertices of the same color ($u'$, and some other $v' \in V(C_{u'})$), while all other vertices are of different colors.
		
		Because each principal point $w$ is bad, we either have $w\in\{u',v'\}$ or the neighbors of $w$ in $C$ are exactly $u'$ and $v'$.
		Altogether, we have at most three principal points.
		Hence, $C$ does not form a convex polygon (as $t\geq 6$ and each point of a convex polygon is a principal point).
		As each polygon has at least three principal points, $C$ has exactly three principal points: $u'$, $v'$ and a principal point $w$ whose neighbors in $C$ are $u'$ and $v'$.
		In particular, the three principal points form a path in $C$.
		\Cref{lem:anthropomorphicPolygon}, stated in \cref{app:monotonicity}, shows that such a geometric cycle contains at most five points, a contradiction.
		Therefore, $C'$ is non-rainbow.
			
		Now assume that $uv$ is in the exterior of the region bounded by $C$.
		Without loss of generality, assume that there are no edges of $H$ that are in the interior of $C$ and do not cross $C$, or, equivalently,
		each noncrossing chord of $C$ in the interior has both endpoints of the same color.
		Consider a triangulation $T$ of $C$, and let $vw$, $ww'$ and $w'w''$ denote the first three edges in $\pi_1$ (which has at least $\lceil\frac{t}{2}\rceil\geq 3$ edges).
		Let $z$ denote the unique vertex in $C$ such that $w$, $w'$, and $z$ form a triangle in $T$ which is in the interior of $C$.
		We remark that the edges of $T$ might differ from the edges of $H$. 
		In particular, $wz$ and $w'z$ are edges of $T$ but (by assumption) can only be in $H$ if they are also in~$C$.
		Since $w$ and $w'$ are of different colors, at least one of $zw$ or $zw'$ is in $H$ and thus~$C$.
		Hence, $z$ is one of $v$ or $w''$.
		As further $t\geq 6$, at least one of $\{zw, zw'\}$ is not in $C$ and thus not in~$H$.
		Consequently, either (if $z=v$) the vertices $v$ and $w'$ have the same color or (if $z=w''$) the vertices $w$ and $w''$ have the same color.
		Hence $C'$ is non-rainbow.

		\item\label[proofcase]{case-2:flip}
		\caseheading{Every edge in $H$ but not in $C$ crosses an edge of $C$.  (Each noncrossing chord of $C$ is not in $H$.)}
		%We can assume that all \colonename edges are in the inside of $C$ and all \coltwoname edges are at the outside of $C$.
		
		Let $T$ be a triangulation of $C$. 
		By assumption, the only edges of the triangulation that are also in $H$ are the edges of $C$.
		%		By assumption, each edge of $T$ that is not in $C$ is monochromatic.
		Let $e$ be an edge in $H$ and not in $C$ and let $T'$ be an arbitrary triangulation of the vertices of $C$ that contains the edge~$e$.
		Such a triangulation exists since $\lvert V(C)\rvert = t \geq 6$.
		Consider the flip sequence $\mathcal{S}$ that flips $T$ to~$T'$.
		We carry out the flips of $\mathcal{S}$ one after the other and stop the first time when an edge $ab$ from $C$ is flipped.
		Let $cd$ denote the resulting edge of that flip.
		Let $T''$ denote the triangulation before the flip.
		
		\begin{figure}
			\centering
			\includegraphics[page=2]{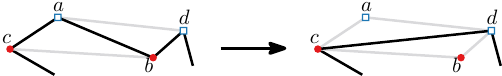}
			\caption{Using a flip to shorten a plane cycle (black edges) by two edges.}
			\label{fig:flip}
		\end{figure}
		
		Because we did not stop earlier, $T''$ is still a triangulation of $C$.
		Hence, each edge of $T''$ that is in $H$ is also in $C$ (by the assumptions of this case).
		In particular, each edge from $\{ac, ad, bc, bd\}$ that is in $H$ is also in $C$.
		Therefore, since $a$ and $b$ have different colors (as $C$ is in $H$), at least one edge from $\{ac, bc\}$ as well as at least one edge from $\{ad, bd\}$ is in $C$.
		In turn, at least one edge from $\{ac, bc\}$ as well as at least one edge from $\{ad, bd\}$ is not in $C$, since $C$ is a cycle and of length at least $6$.
		We may assume, without loss of generality, that $ac$ and $bd$ are in $H$ and the sequence $ca$, $ab$, $bd$ forms a path on $C$.
		Since $ab$ is flippable, the vertices $a$, $c$, $b$, $d$ form a convex quadrilateral.
		Moreover, there are no vertices of $C$ in the interior of this quadrilateral.
		In particular, this also shows that $c$ and $d$ have different colors (one has the same color as $a$ and the other has the same color as~$b$), because the noncrossing chords $ad$ and $bc$ are not in $H$ by the assumptions of this case.
		By replacing the sequence $ca$, $ab$, $bd$ by the edge $cd$ (see \cref{fig:flip}), we therefore obtain a plane cycle $C'$ in $H$ which is two edges shorter than $C$ (but still at least $\lceil\frac{t}{2}\rceil + 1$ edges long since $t \geq 6$).

		We claim that $C'$ is non-rainbow.
		For the sake of a contradiction, assume that $C'$ is rainbow.
		Then all vertices in $C'$ are of different colors.
		So $a$ and $b$ are the only vertices in $C$, whose neighbors in $C$ are of the same color, i.e., $C$ contains two colors which each occur twice.
		$C$ is nonconvex (because of the quadrilateral $a$, $c$, $b$, $d$) and contains at least three principle points, therefore there is a third principle point $z$ in $C$ besides $a$ and~$b$.
		Since $z$ is also part of $C'$, which is rainbow, the neighbors of $z$ have different colors; this is also true if $z = c$ or $z = d$.
		Furthermore, $C_z$ contains $a$ and $b$ and at least one of $c$ or $d$, i.e. $C_z$ is non-rainbow.
		Thus $z$ is a good principal point of $C$, a contradiction.
		Therefore,~$C'$ is non-rainbow.
		\qed
	\end{proofcases}
\end{proof}

\begin{figure}
	\centering
	\includegraphics[page=3]{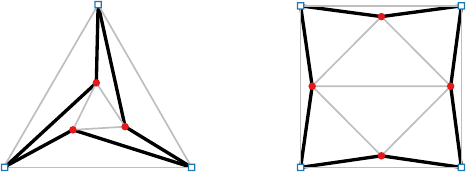}
	\caption{
		Triangulations $T$ of some plane cycles $C$ (black edges, thick).
		\subcaptionheading{Left:} For any $T$ only the edges of $C$ are in $H$. 
		\subcaptionheading{Right:} For any $T$ no edge of $C$ is flippable.
	%It is possible that for any $T$ only the edges of $C$ are in $H$ (left).
	%	It is also possible that for any $T$ no edge of $C$ is flippable (right).
		%		This also holds for any $T$ since for no edge $ab$ of $C$ we find two other vertexs $cd$ for which $abcd$ is a convex quadrilateral with diagonal $ab$.
	}
	\label{fig:remark-flip}
\end{figure}

We remark that the two cases in the proof of \cref{thm:half} are necessary; see \cref{fig:remark-flip}.
In particular, it is possible that in every triangulation of a plane cycle~$C$, only the edges of $C$ are in~$H$.
It is also possible that for every triangulation of~$C$, no edge of $C$ is flippable.
Further note that it is not always possible to shorten an existing plane cycle by replacing a path with three edges with a single edge, as it is done in the second case of the proof of \cref{thm:half}:
an example with a plane 22-cycle is given in \Cref{fig:no-c4-chord} (the underlying geometric graph admits \emph{other} plane 20-cycles).%\todo{AW: For the 20 to make sense, we should also say the 22.}

\begin{figure}[h]
	\centering
	\includegraphics[width=\textwidth]{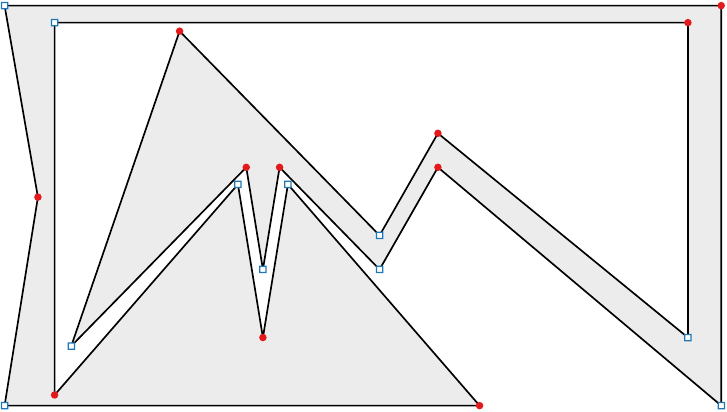}
	\caption{
		A plane cycle $C$ of length~22 with no chord $c$, such that $c$ and a subpath of $C$ forms a plane cycle of length~20.
	}
	\label{fig:no-c4-chord}
\end{figure}

%We conclude this section with a monotonicity result for a restricted type of drawing with $n$ \colonename and $n$ \coltwoname vertices, which generalizes a result of Soukup~\cite[Remark after Theorem~1]{Soukup24}. 
We conclude this section with a monotonicity result for a restricted type of drawing.
In the following, we associate the two partition classes of a complete bipartite graph with two color classes, which we refer to as \colonename and \coltwoname, and color the vertices accordingly.
Our result
%, which 
generalizes the following result of Soukup (rephrased in terms of drawings):
If a geometric complete bipartite graph $H$ with $n$ \colonename vertices and $n$ \coltwoname vertices has a set of \coltwoname vertices that contains  in the interior of its convex hull all \colonename vertices but no \coltwoname vertex, then $H$ contains a plane Hamiltonian cycle~\cite[Remark after Theorem~1]{Soukup24}.

\begin{restatable}{theorem}{ThmNestedColor}\label{thm:nestedColor}
	Let $H$ be a geometric complete bipartite graph with $n$ vertices of each color.
	If there exists a set of \coltwoname vertices that contains  in the interior of its convex hull all \colonename vertices but no \coltwoname vertex, then for each $t\in \{2,\ldots, n\}$, $H$ contains a plane cycle of length~$2t$.
\end{restatable}

\begin{proof}
	Let $B$ be the set of \coltwoname vertices that contains  in the interior of its convex hull all \colonename vertices but no \coltwoname vertex.
	Consider some $t\in \{2,\ldots, n\}$.
	Let $I$ denote the set of \colonename vertices.
	We can assume that $B$ contains all \coltwoname vertices.
	If this is not the case, we can repeatedly remove one \coltwoname vertex not in $B$ and one \colonename vertex in $I$ until either exactly $t$ vertices in $I$ remain or we run out of \coltwoname vertices not in $B$.
	In the former case, also exactly $t$ \coltwoname vertices remain and we apply Soukup's result to find a cycle of length~$2t$.
	In the latter case, we use the following strategy.
	
	Let $b_1, \ldots, b_k$ be the vertices in $B$ in counterclockwise order along the convex hull of $B$ ($k > t$).
	We call the interior of the triangle spanned by $b_{i - 1}, b_i, b_{i + 1}$ the \lightdfn{boundary triangle~$t_i$}.
	Here we use $b_{k + 1} = b_1$ and $b_0 = b_k$, and we also set $t_{k + 1} = t_1$ and $t_{0} = t_k$.
	Note that among the boundary triangles, $t_i$ has a nonempty intersection only with $t_{i - 1}$ and $t_{i + 1}$.
	We call $t_{i} \cap t_{i - 1}$ the \lightdfn{edge-zone}~$z_i$.
	Note that all edge-zones are disjoint.	
	We distinguish two cases:
	
	\begin{proofcases}
		\item\label{case:almost-empty-ear}
		\caseheading{There exists a boundary triangle $t_i$ containing at most one vertex $p$ from $I$.}
		
		Let $B' = B \setminus \{b_i\}$.
		If there is a vertex $p \in t_i$, then we set $p'=p$, otherwise we pick an arbitrary vertex from $I$ as $p'$.
		Next, we set $I' = I \setminus \{p'\}$.
		Because $t_i$ contains no vertices from $I'$, $I'$ lies in the interior of $B'$.
		The resulting set $P' = B' \cup I'$ has size $2k - 2$.
		We can therefore recurse on the geometric complete multipartite graph with vertex set $P'$ until we get to $t$ vertices on the convex hull and apply Soukup's result, or until we end in \cref{case:no-almost-empty-ears}.
		
		\item\label{case:no-almost-empty-ears}
		\caseheading{Every triangle $t_i$ contains two or more vertices from $I$.}
		
		We first argue that in this case every boundary triangle $t_i$ contains exactly two vertices from $I$.
		Let $G$ be a bipartite graph whose vertices represent the boundary triangles $T = \{t_1, \ldots, t_k\}$ as one partition class and the vertices in $I = \{p_1, \ldots, p_k\}$ as the other partition class.
		Whenever a triangle $t_i$ contains a vertex $p_j$, we add the edge $t_ip_j$ to $G$.
		On the one hand, because the triangle $t_i$ has nonempty intersection only with two other triangles, no vertex $p_j$ from $I$ can be in more than two triangles of $T$, so the degree of $p_j$ is at most two and $G$ has at most $2k$ edges.
		On the other hand, by our case assumption, every $t_i$ contains at least two vertices, so the degree of $t_i$ is at least two and $G$ has at least $2k$ edges.
		Thus $G$ has exactly $2k$ edges.
		As a consequence, every vertex $t_i$ or $p_j$ has degree exactly two, and thus every triangle $t_i$ contains exactly two vertices and every vertex $p_j$ from $I$ lies in two triangles simultaneously.
		In particular, each vertex from $I$ lies in an edge-zone.
		
		Given these insights about how the drawing can look locally, we now show that there are only two possible drawings.
		As a first case, assume that there is an edge-zone $z_i$ containing no vertex from $I$.
		Then the two vertices that need to lie in $t_i$ have to lie in $z_{i + 1}$.
		The two vertices in $t_{i + 1}$ are now already in $z_{i + 1}$, and thus $z_{i + 2}$ has to be empty.
		We can repeat this argument until we arrive at $t_{i - 1}$.
		If $k$ is odd, we get a contradiction, since $z_{i - 1}$ would then contain no vertex from $I$ and so would $t_{i - 1}$.
		However, if $k$ is even, we get a feasible drawing, which we call~$\mathcal{C}_1$.
		As a second case, assume that every edge-zone contains at least one vertex from $I$.
		As the edge-zones are disjoint, we have that each edge-zone contains exactly one vertex from $I$.
		This gives another valid drawing, which we call $\mathcal{C}_2$.
		Both drawings are depicted in \cref{fig:nestedColor-second-case-configurations}.
		
		\begin{figure}
			\centering
			\includegraphics[width=\textwidth]{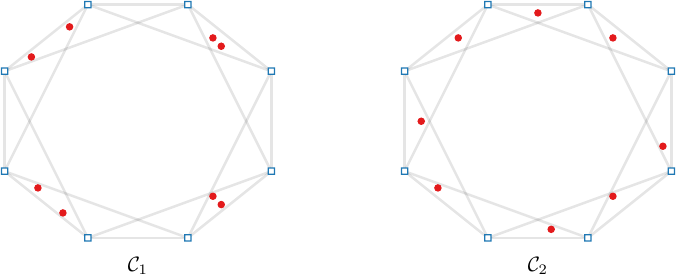}
			\caption[The two possible drawings if every triangle contains two or more vertices]
			{The two possible drawings in the proof of \cref{thm:nestedColor}, \cref{case:no-almost-empty-ears}.}
			\label{fig:nestedColor-second-case-configurations}
		\end{figure}

		We can now finally show the existence of a plane cycle of length $2t$, by deriving it from a plane cycle $C_{2k}$ of length $2k$.
		For drawing $\mathcal{C}_1$, a possible $C_{2k}$ is given by connecting the vertices as follows.
		For every edge-zone $z_i$ that contains two vertices $p_1, p_2 \in I$, assume without loss of generality that $b_{i - 1}$ sees the vertices $b_i$, $p_1$, $p_2$ and $b_{i - 2}$ in counterclockwise order.
		We connect both $b_{i - 1}$ and $b_i$ to $p_1$ and both $b_{i - 1}$ and $b_{i - 2}$ to $p_2$ to get $C_{2k}$; see \cref{fig:aphc-in-second-configuration}.
		A plane cycle $C_{2k - 2}$ of length $2k - 2$ can be obtained from $C_{2k}$ by picking a nonempty edge-zone $z_i$ and skipping the vertex $b_{i - 1}$ and one vertex in $z_i$.
		For smaller cycle lengths, we prune $\mathcal{C}_1$ by repeatedly removing the vertices in some nonempty $z_i$ together with $b_i$ and $b_{i - 1}$ until we reach a drawing $\mathcal{C}'_1$ of the same type as $\mathcal{C}_1$, with $2t$ or $2t + 2$ vertices.
		We can then find the desired cycle with the strategy that defined $C_{2k}$ or $C_{2k - 2}$ above.
		
		\begin{figure}
			\centering
			\includegraphics{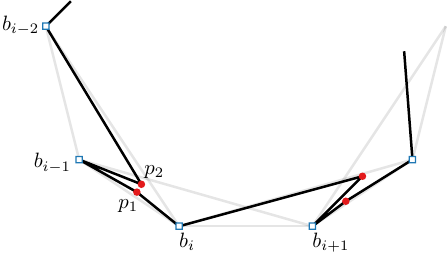}
			\caption[A plane cycle of length 2k in case there are empty edge-zones]
			{A plane cycle of length $2k$ (black) in the drawing $\mathcal{C}_1$.}
			\label{fig:aphc-in-second-configuration}
		\end{figure}
		
		For $\mathcal{C}_2$, a possible $C_{2k}$ is given by connecting each vertex $b_i$ to the vertices within the boundary triangle $t_i$.
		For smaller cycles we can again prune $\mathcal{C}_2$, this time by simply taking $2t$ consecutive vertices in $C_{2k}$.
		By doing this, we introduce only one new edge, which lies in the interior of $C_{2k}$ and thus does not cross $C_{2k}$.\qed
	\end{proofcases}
\end{proof}

%\begin{proof}[sketch, full proof in \cref{app:ssec:thmnestedcolor}]
%	The arguments in \cref{app:ssec:thmnestedcolor} show how to reduce $H$ to a subgraph from one of the following categories:
%	Either we are left with $t$ vertices of each color satisfying the assumptions of Soukup's result, which then gives a plane cycle of length~$2t$.
%	Or we reach one of the geometric complete bipartite graphs illustrated in \cref{fig:nestedColor-second-case-configurations}.
%	In the latter case we prune that graph to $t$ vertices of each color in a systematic way and show that it contains a plane cycle of length~$2t$.\qed
%\end{proof}	
%	
%	\begin{figure}
%		\centering
%		\includegraphics[page=2]{separated-knn-second-case-configurations}
%		\caption[Two possible drawings]
%		{Two possible drawings in the proof of \cref{thm:nestedColor}.}
%		\label{fig:nestedColor-second-case-configurations}
%	\end{figure}

%%%%%%%%%%%%%%%%%%%%%%%%%%%%%%%%%%%%%%
%%%%%%%%%%%%%%%%%%%%%%%%%%%%%%%%%%%%%%
%%%%%%%%%%%%%%%%%%% FPT section
%%%%%%%%%%%%%%%%%%%%%%%%%%%%%%%%%%%%%%
%%%%%%%%%%%%%%%%%%%%%%%%%%%%%%%%%%%%%%

\section{An FPT Algorithm with respect to the number of interior vertices}\label{sec:FPT}

In this section we consider geometric complete bipartite host graphs $H$ with $n$ \colonename vertices and $n$ \coltwoname vertices.
%We refer to edges in $H$ as \lightdfn{bichromatic} edges, and to edges not in $H$ as \lightdfn{monochromatic} edges.\todo{check mono/bi edges sentence}
The vertices of $H$ that lie in the interior of the convex hull of $V(H)$ are called the \lightdfn{interior vertices} and the remaining vertices are called \lightdfn{boundary vertices}.
We describe an algorithm that checks whether $H$ contains a plane Hamiltonian cycle in time $O(n\log n + n k^2)+O(k^{5k})$, where $k$ is the number of interior vertices.

Let $I$ denote the set of interior vertices and let $B= V(H)\setminus I$ denote the set of boundary vertices of $H$.
Let $w_1,\ldots,w_m$ denote the boundary vertices in counterclockwise order where $w_1$ is chosen arbitrarily and indices are used modulo $m$ (in particular, $w_0=w_m$).
The case $k\leq 1$ needs to be handled separately and is covered by the following observation.

\begin{observation}\label{obs:FTPnearConvex}
	If $H$ has at most one interior vertex, then $H$ contains a plane Hamiltonian cycle if and only if there is at most one pair of vertices $w_i$, $w_{i+1}$ of the same color.
\end{observation}

This can be easily checked by computing the convex hull of $V(H)$ and checking the boundary vertices in counterclockwise order in $O(n\log n)$ time~\cite{BCKO08}.
We assume $k\geq 2$ for the rest of the section.

\subsection{Structural Considerations}

First, we describe some structural insights that will used by the algorithm.
The main idea is to split $B$ into $O(k^2)$ parts, called critical arcs, such that it is sufficient to consider only few vertices within each part.
A vertex $w_i\in B$ is \lightdfn{critical} if $w_i$ and $w_{i-1}$ are of the same color (\enquote{first kind}) or there is a line through two vertices from $I$ that separates $w_i$ from $w_{i-1}$ (\enquote{second kind}).
\Cref{fig:criticalPoints} shows an example.
Let $S=\{w_{a_1},\ldots,w_{a_s}\}$ denote the set of critical vertices where $a_1 < a_2 < \cdots < a_s $.
For each critical vertex $w_{a_i}$  the \lightdfn{critical arc at $a_i$} is the set $\{w_{a_i},\ldots,w_{a_{i+1}-1}\}$ of boundary vertices between $w_{a_i}$ and the next critical vertex (including the first and excluding the latter critical vertex).
Observe that the critical arcs are pairwise disjoint and form a partition of $B$.
%Let $a'_i=a_{i+1}-1$ denote the index of the last vertex in the critical arc at $a_i$ (possibly $a_i=a'_i$).
The following \lcnamecref{lem:criticalArcs} describes the properties of critical arcs needed later.
%The proof can be found in the \cref{app:ssec:criticalArcsLem}.
For given vertices $u$ and $v$ we denote by $uv$ the line segment between $u$ and $v$, regardless of whether $uv$ is an edge of $H$ or not.

\begin{restatable}{lemma}{arcLem}\label{lem:criticalArcs}
	Let $u$, $v\in I$, let $A$ and $\bar{A}$ be distinct critical arcs, and let $w$, $w'\in A$.
	\begin{arcproperties}
		\item Each critical arc is alternately colored along $B$.\label{enum:criticalArcAlternately}
		
		\item The vertices $w$ and $w'$ lie on the same side of the line through $uv$.\label{enum:criticalOrderType}
		
		\item Either all line segments between $u$ and $A$ cross all line segments between $v$ and $\bar{A}$, or none of these segments cross.\label{enum:criticalDistinctCrossing}
		
		\item If $w$ comes before $w'$ in a counterclockwise traversal of $A$, then $uw$ and $vw'$ do not cross if and only if $A$ lies to the right of the line through $u$ in direction of $v$.\label{enum:criticalSameCrossing}
	\end{arcproperties}
\end{restatable}
\begin{proof}
	\Cref{enum:criticalArcAlternately,enum:criticalOrderType,enum:criticalDistinctCrossing} follow directly from the definition of critical vertices and arcs.
	To prove \cref{enum:criticalSameCrossing}, consider the lines $\ell_{uv}$, through $u$ and $v$, and $\ell_{ww'}$, through $w$ and $w'$.
	By the definition of critical arcs, $\ell_{uv}$ does not intersect $A$.
	Hence, $A$ lies completely on one of the two sides of $\ell_{uv}$.
	In particular, $w$ and $w'$ lie on the same side.
	This also shows that $u$ and $v$ lie on the same side of $\ell_{ww'}$, as $w$ and $w'$ are boundary vertices of $H$.
	If $w$ comes before $w'$ in a counterclockwise traversal of~$A$, then rotating $\ell_{uv}$ about $u$ meets $w$ before it meets $w'$.
	From this we see that $uw$ and $vw'$ do not cross if and only if $A$ lies to the right of $\ell_{uv}$.
	\qed
\end{proof}

\begin{figure}
	\centering
	\includegraphics{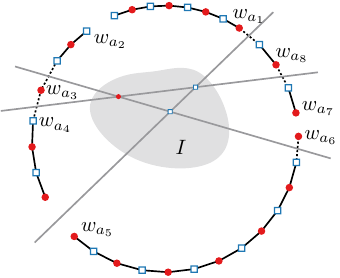}
	\caption[The critical vertices of a geometric complete bipartite host graph]
	{
		The vertices $w_{a_2}$, $w_{a_5}$, and $w_{a_7}$ are critical vertices of the first kind, $w_{a_1}$, $w_{a_3}$, $w_{a_4}$, $w_{a_6}$, and $w_{a_8}$ are critical vertices of the second kind.
		Critical arcs are shown by solid black edges.
	}
	\label{fig:criticalPoints}
\end{figure}

We also use the following properties of plane Hamiltonian cycles which are straightforward to check.

\begin{observation}\label{lem:propertiesHC}
	Let $C$ be a plane Hamiltonian cycle in $H$.
	\begin{properties}
%		\item For each edge $w_iw_j$ in $C$ we have $j= i\pm 1$.\label[property]{enum:HCnochord} \TODO{Where is this used?}
		\item There are at most $k$ critical vertices of the first kind.\label[property]{enum:HCcriticalpoint}
		\item The cycle $C$ visits the boundary vertices in a cyclic order along $B$.\label[property]{enum:HCcyclicorder}
	\end{properties}
\end{observation}

%We now describe the desired algorithm.
The idea of the algorithm is to iterate over all spanning cycles of the interior vertices $I$ (there are $O(k!)$ many) and check if
one of them can be extended to the whole graph $H$. 
Observe that the assumption $k\geq 2$ ensures that there are at least two critical arcs.

An \lightdfn{initial cycle} $F$ of $I$ is a not-necessarily plane, geometric directed spanning cycle with vertex set $I$ (that may use oriented edges from $H$ as well as edges not in $H$), together with a nonempty subset $G_F$ of the edges of $F$, called \lightdfn{gaps}, such that the edges of $F$ not in $G_F$ (called \lightdfn{fixed edges}) are in $H$ and do not cross each other.
The gaps represent those edges in $F$ that should be replaced with longer paths using boundary vertices from~$B$. 
In particular, every gap $uv$ will be replaced by a path $u,w_i,w_{i+1},\ldots,w_{i+j},v$, where $w_i$, $w_{i+1},\ldots,w_{i+j}\in B$.
Our definition includes that in case $k=2$, $F$ is the directed cycle with just two edges between the two vertices in $I$.
%Note that by~\cref{lem:propertiesHC} each such path needs to form an interval along $B$.
For technical reasons, if two gaps are adjacent in~$F$, we insert into $F$ a fixed zero-length dummy edge at the common endpoint and color the new vertices with the color of the replaced vertex.
Note that dummy edges have both endpoints of the same color.
%We also adjust $H$ to include all bichromatic edges adjacent to the endpoints of dummy edges.
We hence can assume for the rest of the section that each endpoint of a gap is incident to only this one gap.
We call the endpoints of gaps in $G_F$ \lightdfn{gap vertices}. % and denote the gap vertices as 
%$u_1,\ldots u_g$ in the order as they appear in $F$ (for some arbitrary starting point).

To extend $F$ to $B$, for every gap vertex $u$ we select a critical arc $A(u)$, and try to place a (new) neighbor for $u$ in $A(u)$. 
The crucial observation is that we do not need to check all individual vertices from the selected critical arcs to decide whether $F$ can be extended to $B$ in this way.
Instead, we find necessary and sufficient conditions for the selection of arcs alone to be \enquote{feasible}.
This gives a running time per selection that is independent of $n$ as needed.

%Assume we have selected a critical arc $A_i$ for each gap vertex $u_i$, where each arc might be selected multiple times.
%The selection defines an ordered sequence $\mathcal{A} = A_{1},\ldots,A_{g}$ of the arcs, where $A_i= A(u_i)$.
%By \cref{enum:HCcyclicorder} (\cref{lem:propertiesHC}) $\mathcal{A}$ describes a counterclockwise ordering of the selected arcs along $B$ (where repeatedly selected arcs form an interval in the ordering).
%
%Let $A_{1},\ldots,A_{g}$ denote a counterclockwise ordering of the selected arcs along $B$ (where repeatedly selected arcs form an interval in the ordering), such that if $A_{i}=A_{i+1}$ (and $u_iu_{i+1}$ is not a dummy edge pair), then $A_{i}$ lies to the right of the line through $u_i$ in direction $u_{i+1}$, where $u_i$ and $u_{i+1}$ are the gap vertices with $A_i=A(u_i)$ and $A_{i+1}=A(u_{i+1})$.
Assume we have selected a critical arc $A(u)$ for each gap vertex $u$, where each arc might be selected multiple times.
We call the selection of critical arcs \lightdfn{feasible} if a labeling $u_1,\ldots,u_g$ of the gap vertices along $F$ exists, such that the following conditions hold (where $A_i=A(u_i)$):
\begin{conditions}
	\item 
	The selected arcs $A_1,\ldots,A_g$ appear counterclockwise along $B$ in this order.
	Each repeatedly selected arc forms an interval in this ordering.
	If $A_{i}=A_{i+1}$ for some $i$ (and $u_iu_{i+1}$ is not a dummy edge pair), then $A_{i}$ lies to the right of the line through $u_i$ in direction $u_{i+1}$.
\label[condition]{enum:feasibleOrder}
	
	\item For each $i$, the line segments between $u_i$ and $A_{i}$ do not cross fixed edges from $F$.\label[condition]{enum:feasibleCrossF}
	
	\item If $A_{i} \neq A_{j}$, then the line segments between $u_i$ and $A_i$ do not cross line segments between $u_j$ and $A_{j}$.\label[condition]{enum:feasibleCrossArcs}
	
	\item If $u_iu_{i+1}$ is a gap, then there is no critical vertex of the first kind between $A_i$ and $A_{i+1}$, that is, the union of all critical arcs from $A_i$, counterclockwise along $B$ up to $A_{i+1}$ (including the first and the latter) is alternately colored.\label[condition]{enum:feasibleAlternatingPaths}
	
	\item If $u_iu_{i+1}$ is not a gap and both vertices are of the same color, then \mbox{$A_i\neq A_{i+1}$}.\label[condition]{enum:feasibleMonoGap}
	
	\item If $u_iu_{i+1}$ is not a gap and $A_i\neq A_{i+1}$, then $A_i$ and $A_{i+1}$ are adjacent along~$B$, the (counterclockwise) last vertex in $A_i$ is of opposite color to~$u_i$, and the first vertex in $A_{i+1}$ is of opposite color to $u_{i+1}$.\label[condition]{enum:feasibleNongap}
	
	\item For each selected critical arc $A$ we have $\lvert A\rvert \geq g(A) - \mathrm{mc}(A) + \varepsilon(A),$
	where $g(A)$ is the number of gap vertices $u$ with $A(u)=A$, $\mathrm{mc}(A)$ is the number of gaps $uv$ with $uv\not\in H$ and $A(u)=A(v)=A$, and $\varepsilon(A)=1$ if the (counterclockwise) first vertex in $A$ is of the same color as the (counterclockwise) first gap vertex with $A(u)=u$ and $\varepsilon(A)=0$ otherwise.\label[condition]{enum:feasibleSize}
\end{conditions}

The following \lcnamecref{lem:FPTfeasible} shows that these conditions are sufficient and necessary.
We say that $F$ and the selected arcs are \lightdfn{Hamiltonian} if there is a plane Hamiltonian cycle $C$ in $H$ such that a counterclockwise traversal of $C$ visits the vertices in $I$ in the order given by $F$ and for each gap vertex $u$ there is an edge $uw$ in $C$ with $w\in A(u)$.

\begin{restatable}{lemma}{feasibleSelection}\label{lem:FPTfeasible}
	An initial cycle and a selection of critical arcs are Hamiltonian if and only if the selection is feasible.
\end{restatable}
\begin{proof}
	Let $F$ denote an initial cycle and let $A:u\mapsto A(u)$ denote a selection of critical arcs.
	
	First assume that $F$ and the selected critical arcs are Hamiltonian and let $C$ denote a corresponding plane Hamiltonian cycle in $H$.
	We prove that the selection is feasible by individually checking the feasibility conditions.
	Consider a counterclockwise traversal of $C$ such that the second vertex in this traversal is some arbitrary critical vertex.
	Let $u_1,\ldots,u_g$ denote the gap vertices of $F$ as they appear in this traversal of $C$ (which visits all vertices from $F$ as it is Hamiltonian), where we introduced dummy edges for adjacent gaps as described above.	
	Let $A_i=A(u_i)$ and observe that $C$ contains an edge between $u_i$ and $A_i$ by definition.
	\begin{conditions}%[label=(C\arabic*), wide=0pt, leftmargin=\parindent]
		\item By \cref{enum:HCcyclicorder}, the sequence $A_1,\ldots,A_g$ of selected critical arcs appears counterclockwise along $B$, where each repeatedly selected arc forms an interval.
		Indeed, the latter is clearly true for all arcs (for any counterclockwise traversal of $C$) except the ones selected for $u_1$ and $u_g$, which might be identical.
		Because the second vertex in the traversal of $C$ which defines the order $u_1,\ldots,u_g$ is some critical vertex, we have $A_1=A_g$ only if all selected arcs are equal and hence they form an interval as well.
		Finally, if $A_i=A_{i+1}$ and $u_iu_{i+1}$ is not a dummy edge, then $C$ contains edges $u_iw$ and $u_{i+1}w'$ with $w$, $w'\in A_i$ which do not cross, and hence, by \cref{enum:criticalSameCrossing}, $A_i$ lies to the right of the straight line through $u_i$ in direction of $u_{i+1}$.
		\item As $C$ is plane, Hamiltonian, and contains all fixed edges from $F$, there is an edge $u_iw$ in $C$, with $w\in A_i$, not crossing any fixed edge from $F$.
		Consider an arbitrary fixed edge $e$ and the straight line through $u_i$ and an endpoint of $e$ such that $w$ and $e$ lie on opposite sides of the line.
		By \cref{enum:criticalOrderType}, each vertex from $A_i$ lies on the same side of this line as $w$ lies (hence, on the side opposite to $e$).
		Hence, all the line segments between $u_i$ and $A_i$ do not cross $e$.
		
		\item As $C$ is plane, its edge between $u_i$ and $A_i$ and its edge between $u_j$ and $A_{j}$ do not cross each other.
		By \cref{enum:criticalDistinctCrossing}, all the line segments between $u_i$ and $A_i$ do not cross the line segments between $u_j$ and $A_j$.
		
		\item If $u_iu_{i+1}$ is a gap, then, by \cref{enum:HCcyclicorder}, $C$ contains a path from some vertex $w\in A_i$ to some vertex $w'\in A_{i+1}$ which visits all boundary vertices between $w$ and $w'$ along $B$.
		In particular, there is no critical vertex of the first kind between $A_i$ and $A_{i+1}$.
		
		\item If $u_iu_{i+1}$ is not a gap, then $C$ contains edges $u_iw_j$ and $u_{i+1}w_{j+1}$, with $w_{j}\in A_i$ and $w_{j+1}\in A_{i+1}$ by \cref{enum:HCcyclicorder} ($w_j$ and $w_{j+1}$ are adjacent in $B$).
		If $u_i$ and $u_{i+1}$ are of the same color, then $w_j$ and $w_{j+1}$ are of the same color.
		Hence, $w_{j+1}$ is a critical vertex of the first kind.
		Thus, $w_j$ and $w_{j+1}$ are in distinct critical arcs, that is, $A_i\neq A_{i+1}$.
		
		\item If $u_iu_{i+1}$ is not a gap, then $C$ contains edges $u_iw_j$ and $u_{i+1}w_{j+1}$, with $w_{j}\in A_i$ and $w_{j+1}\in A_{i+1}$ by \cref{enum:HCcyclicorder} ($w_j$ and $w_{j+1}$ are adjacent in $B$).
		So, if $A_i\neq A_{i+1}$, then these critical arcs are adjacent along $B$.
		Moreover, $w_i$ is the last vertex in $A_i$ (and hence this vertex is of opposite color to $u_i$ as $C$ is in $H$) and $w_{i+1}$ is the first vertex in $A_{i+1}$ (and hence this vertex is of opposite color to $u_{i+1}$).
		
		\item For each gap vertex $u$ there is an edge between $u$ and $A(u)$ in $C$.
		As $C$ is a cycle, each vertex in $B$ has in $C$ at most two neighbors from $I$.
		If a vertex in $B$ has two distinct neighbors $u$, $v\in I$ in $C$, then $u$ and $v$ are of the same color and $A(u)=A(v)$.
		Since \cref{enum:feasibleMonoGap} holds, $uv\in G_F$.
		This already shows that for every arc $A$ we have $\lvert A\rvert \geq g(A)-\mathrm{mc}(A)$.
		If the  first vertex $w$ in $A$  (in counterclockwise direction) is of the same color as the first gap vertex (in the order $u_1,\ldots, u_g$) with $A(u)=A$, then $w$ has, in $C$, no neighbor from $I$ and hence $\lvert A\rvert \geq g(A)-\mathrm{mc}(A)+1$ in this case.
	\end{conditions}
	
	Now assume that the selection is feasible.
	We shall prove that $F$ and the selection are Hamiltonian.
	To this end we give a construction of a suitable plane Hamiltonian cycle of $H$.
	For this we process the selected arcs in counterclockwise order.
	%	 along the order $u_1,\ldots,u_g$ given by \cref{enum:feasibleOrder}.
	For each selected arc $A$, the gap vertices with selected arc $A$ form an interval $u_j,\ldots,u_{j+t}$ in the order of gap vertices given by \cref{enum:feasibleOrder}.
	For each $u_i$, with $j\leq i \leq j+t$, we determine a neighbor $w_{x_i}$ in $A$.
	We choose the vertices $w_{x_j},\ldots,w_{x_{j+t}}$ in ascending order with the following strategy:
	First assume that $t=0$, so $A$ is selected for $u_j$ only.
	If $u_{j-1}u_{j}$ is not a gap, let $w_{x_j}$ be the (counterclockwise) first vertex in $A$.
	If $u_{j-1}u_{j}$ is a gap, then let $w_{x_j}$ be the (counterclockwise) last vertex in $A$.
	In both cases, $w_{x_j}$ is of opposite color to $u_j$ by \cref{enum:feasibleNongap}.
	
	Now assume that $t>0$.
	If $u_{j-1}u_{j}$ is not a gap (see \cref{fig:assign}(ii)), let $w_{x_j}$ be the (counterclockwise) first vertex in $A$ as above.
	Again, $w_{x_j}$ is of opposite color to $u_j$ by \cref{enum:feasibleNongap}.
	If $u_{j-1}u_{j}$ is a gap (see \cref{fig:assign}(iii--iv)) we choose $w_{x_j}$ as the first or second vertex in $A$ such that $w_{x_j}$ is of opposite color to $u_j$.
	Such a vertex exists by \cref{enum:feasibleSize,enum:criticalArcAlternately}.
	Consider $i$ with $j< i<j+t$ and assume that $x_j,\ldots,x_{i-1}$ are already chosen.
	If $u_{i-1}u_{i}$ is not a gap (see \cref{fig:assign}(v)), then choose $x_{i}=x_{i-1}+1$.
	Note that the colors of $u_{i-1}$ and $u_i$ differ due to \cref{enum:feasibleMonoGap} and the colors of $u_{i-1}$ and $w_{x_{i-1}}$ differ by construction.
	Hence, the colors of $w_{x_{i}}$ and $u_i$ differ as well.
	Otherwise (see \cref{fig:assign}(vii--viii)) choose $x_{i}\in \{x_{i-1},x_{i-1}+1\}$ such that $w_{x_{i}}$ is of opposite color to $u_{i}$.
	Again, this choice is possible by \cref{enum:feasibleSize,enum:criticalArcAlternately}.
	Finally consider $i=j+t$ and assume that $x_j,\ldots,x_{j+t-1}$ are already chosen.
	If $u_{j+t-1}u_{j+t}$ is not a gap, choose $x_{j+t}=x_{j+t-1}+1$.
	Again, this choice is possible by \cref{enum:feasibleSize,enum:criticalArcAlternately}.
	If $u_{j+t-1}u_{j+t}$ is a gap, choose $w_{x_{j+t}}$ as the (counterclockwise) last vertex in $A$.
	Then $w_{x_{j+t}}$ is of opposite color to $u_{j+t}$ by \cref{enum:feasibleNongap}, as $u_{j+t}u_{j+t+1}$ is a gap and $A(u_{j+t})\neq A(u_{j+t+1})$.
	
	\begin{figure}
		\centering
		\includegraphics{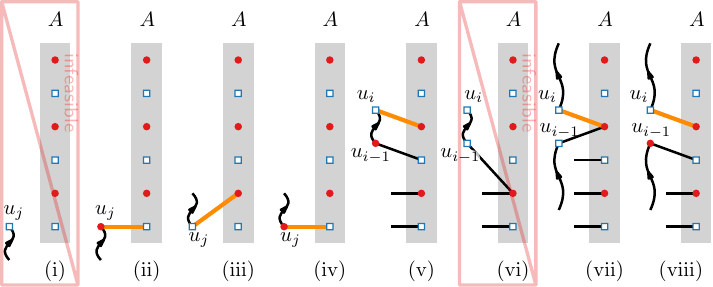}
		\caption[Assigning gap vertices to vertices in their selected arcs]
		{
			Assigning a gap vertex to a vertex in its selected arc $A$.
			The assignment is indicated by orange thick edges, $A$ is shown counterclockwise along $B$ from bottom to top, $F$ is shown by the arrowed black lines.
			Parts (i)--(iv) depict the assignment of the first vertex $u_j$ for $A$,  parts (v)--(viii) the assignment of  subsequent vertices $u_i$ with $j<i<j+t$.
			Note that the situations (i) and (vi) cannot occur in feasible assignments because they violate \cref{enum:feasibleNongap} and \cref{enum:feasibleMonoGap}, respectively.
		}
		\label{fig:assign}
	\end{figure}
	
	So far we chose a neighbor $w_{x_i}$ for each gap vertex $u_i$ such that the edge $w_{x_i}u_i$ is in $H$.
	Let $P_i$ denote the directed path from $w_{x_i}$ to $w_{x_{i+1}}$ counterclockwise along~$B$.
	Let $C$ denote the walk obtained from $F$ by replacing each gap $u_iu_{i+1}\in G_F$ by $P_i$ (keeping the gap vertices, $P_i$ is placed between $u_i$ and $u_{i+1}$).
	We claim that $C$ is a plane Hamiltonian cycle in $H$.
	
	To see that $C$ is in $H$, observe that edges from $F$ in $C$ are exactly the fixed edges, which are by definition in $H$, since $F$ is a initial cycle.
	Moreover, each edge in $C$ between $I$ and $B$ is in $H$ due to the choice of $w_{x_i}$'s above.
	Finally, if $P_i$ is in $C$, then $u_iu_{i+1}\in G_F$ and hence $P_i$ is in $H$ by \cref{enum:feasibleAlternatingPaths}, as the selection of arcs is feasible and $w_{x_i}\in A_i$ and $w_{x_{i+1}}\in A_{i+1}$ by construction.
	
	To see that $C$ is Hamiltonian, observe that each vertex from $I$ is contained in $C$ exactly once, because $F$ is a spanning cycle of $I$.
	It remains to show that the paths $P_i$ are disjoint and cover $B$ entirely.
	The paths $P_i$ are disjoint, because the sequence $u_1,\ldots,u_g$ of gap vertices alternates between the first and last endpoint of a gap (as it follows the direction of~$F$).
	Therefore, also the sequence $w_{x_1},\ldots,w_{x_g}$ alternates between the first and last vertices of the paths~$P_i$.
	Moreover, due to the choice of $x_1,\ldots,x_g$ above, for each $i$ where $u_iu_{i+1}$ is not a gap, $x_{i+1}=x_{i}+1$ holds.
	This shows that each path $P_i$ in $C$ starts right after the previous path ended and ends directly before the next path starts.
	So, $C$ is Hamiltonian.
	
	To see that $C$ is plane, consider the three types of edges in $C$: edges within~$I$, edges between $I$ and $B$, and edges within $B$.
	The latter edges do not cross any edges as they are on the boundary of the convex hull.
	The edges within $I$ are the fixed edges from $F$ which do not cross each other as $F$ is an initial cycle, and which do not cross edges between $I$ and $B$ by \cref{enum:feasibleCrossF} as the selection is feasible.
	Finally, consider edges $u_iw_{x_i}$ and $u_j w_{x_j}$, with $i<j$, between $I$ and $B$.
	If $A_i\neq A_j$, then these edges do not cross by \cref{enum:feasibleCrossArcs}.
	If $A_i=A_j$, then $A_i$ lies to the right of the straight line through $u_i$ in direction of $u_j$, by \cref{enum:feasibleOrder}.
	Due to the choice of $x_i$ and $x_j$ above, we have $x_i\leq x_j$.
	%	As $w_{x_1}, \ldots, w_{x_m}$ are labeled counterclockwise along $B$, we have that $u_i$ and $u_j$ are to the left of the straight line through $w_{x_i}$ in direction $w_{x_{i+1}}$.
	Hence, $u_iw_{x_i}$ and $u_j w_{x_j}$ do not cross by \cref{enum:criticalSameCrossing}.
	So $C$ is plane.
	Altogether we have that $F$ and the selection of critical arcs are Hamiltonian if and only if the selection is feasible.
	\qed
\end{proof}

\cref{lem:FPTfeasible} implies that $H$ contains a plane Hamiltonian cycle if and only if there is an initial cycle of $I$ with a feasible selection of critical arcs.

%The importance of initial cycles and feasible critical arcs comes from the following characterization, which is proven in \cref{app:ssec:FPTgeneralHC}.
%\begin{restatable}{lemma}{FPTgeneralHC}\label{lem:FPTgeneralHC}
%	The geometric complete bipartite graph 
%\end{restatable}

%An algorithm checking the feasibility conditions is provided by \cref{lem:FPTalgorithm} given in \cref{app:FPT}.

\subsection{The Algorithm}

The following two lemmas describe how to compute the critical arcs and how to check whether a selection of critical arcs is feasible.

\begin{lemma}\label{lem:FPTcomputeCritical}
	The set of critical arcs can be computed in $O(nk^2 + n\log n)$ time.	
\end{lemma}
\begin{proof}
	First compute the convex hull of $V(H)$ and label the boundary vertices as $w_1,\ldots,w_m$ in $O(n\log n)$ time~\cite{BCKO08} in order of their appearance.
	For each $w_i$, check whether $w_{i}$ and $w_{i-1}$ are of the same color and whether some line through two vertices from $I$ separates $w_i$ from $w_{i-1}$.	
	This can be done in $O(nk^2)$ time.
	\qed
\end{proof}

\begin{lemma}\label{lem:FPTalgorithm}
	There is an algorithm checking whether a selection of critical arcs for an initial cycle is feasible in time $O(k^2)$, provided that $B$ contains at most $k$ critical vertices of the first kind.
\end{lemma}
\begin{proof}
	We are given an initial cycle $F$ of $I$ and a selection of critical arcs $A$.
	The algorithm checks whether the selection of critical arcs is feasible in the following steps.
	\begin{conditions}
		\item For each gap vertex $u$ consider the sequence $u=u_1,\ldots,u_g$ of gap vertices starting at $u$ and following the direction of $F$.
		Then check whether the corresponding sequence of selected arcs $A_1,\ldots,A_g$ appears counterclockwise along $B$ and whether, in case $A_i=A_{i+1}$, the arc $A_i$ lies to the right of the line through $u_i$ in direction $u_{i+1}$.
		By assumption on the number of critical vertices of the first kind, there are at most $k+k^2$ many critical arcs.
		Hence, this check needs $O(k^2)$ time (as $g\leq 2k$, where the factor $2$ accounts for the dummy edges).
		
		\item Consider each $i$ with $1\leq i \leq g$ and each fixed edge $e$ in $F$.
		By \cref{enum:criticalOrderType}, either all line segments between $u_i$ and $A_i$ cross $e$, or none do.
		So it suffices to check whether the line segment between $u_i$ and the (counterclockwise) first vertex in $A_i$ crosses $e$.
		This needs time $O(k^2)$ in total.
		
		\item Consider each pair $i$, $j$, with $1\leq i < j \leq g$.
		By \cref{enum:criticalDistinctCrossing}, either all line segments between $u_i$ and $A_i$ cross all line segments between $u_j$ and $A_{j}$, or none do.
		So if $A_i\neq A_{j}$, it suffices to check whether the line segment between $u_i$ and the (counterclockwise) first vertex in $A_i$ crosses the line segment between $u_j$ and the first vertex in $A_{j}$.
		This needs time $O(k^2)$ in total.
		
		\item Consider each $i$ with $1\leq i \leq g$.
		If $u_iu_{i+1}\in G_F$, then check whether there is a critical vertex of the first kind between $A_i$ and $A_{i+1}$.
		This can be done in $O(k^2)$ time in total by checking the critical vertices along $B$.
		
		\item Consider each $i$ with $1\leq i \leq g$.
		If $u_iu_{i+1}\not\in G_F$ and both vertices are of the same color, then check $A_i\neq A_{i+1}$.
		This needs $O(k)$ time in total.
		
		\item Consider each $i$ with $1\leq i \leq g$.
		If $u_iu_{i+1}\not\in G_F$ and $A_i\neq A_{i+1}$, check whether $A_i$ and $A_{i+1}$ are adjacent along $B$, the (counterclockwise) last vertex in $A_i$ is of opposite color to $u_i$, and the first vertex in $A_{i+1}$ is of opposite color to $u_{i+1}$.
		This needs $O(k)$ time in total.
		
		\item Determine for each selected arc $A$ the numbers $g(A)$, $\mathrm{mc}(A)$, and $\varepsilon(A)$ by iterating the sequence $u_1,\ldots,u_g$ once.
		Then check whether $\lvert A\rvert \geq g(A)-\mathrm{mc}(A)+\varepsilon(A)$ holds for each~$A$.
		This needs $O(k)$ time.
		\qed
	\end{conditions}
\end{proof}

The final algorithm for the case $k\geq 2$ works as follows.
We compute the critical arcs in $O(n\log n + n k^2)$ time as described in \cref{lem:FPTcomputeCritical}.
%
%via \cref{lem:FPTcomputeCritical} from \cref{app:FPT} and 
We check whether there are more than $k$ critical vertices of the first kind (cf.~\cref{enum:HCcriticalpoint} from Observation~\ref{lem:propertiesHC}).
If not, we generate all $O(k!)$ directed spanning cycles of~$I$.  
We then iterate, for each cycle, over all possible sets of gaps ($O(2^k)$ many) and check if this forms an initial cycle (fixed edges are noncrossing and in~$H$) in $O(k^2)$ time for each.
For each initial cycle, we separate adjacent gaps by dummy edges in $O(k)$ time, iterate over all $O(k^{4k})$ selections of critical arcs, and check feasibility for each selection in $O(k^2)$ time as described in \cref{lem:FPTalgorithm}.
%via \cref{lem:FPTalgorithm} from \cref{app:FPT}.
Altogether, this needs $O(n\log n + n k^2) + O(k^{5k})$ time, as desired.
%Correctness follows from \cref{enum:HCcriticalpoint} together with the fact that a geometric complete bipartite graph contains a plane Hamiltonian cycle if and only if there is an initial cycle of $I$ with a feasible selection of critical arcs (cf. \cref{lem:FPTgeneralHC}), which is stated and proven in \cref{app:ssec:FPTgeneralHC}).
%Correctness is provided by \cref{enum:HCcriticalpoint} and \cref{lem:FPTgeneralHC}. 
%Correctness is provided by \cref{enum:HCcriticalpoint} and \cref{lem:FPTgeneralHC} from \cref{app:FPT}. 
%We conclude with the following
%\lcnamecref{thm:FPT}.
In summary, the following \lcnamecref{thm:FPT} holds.
%A correctness proof is given in \cref{app:ssec:FPTthm}.

\begin{restatable}{theorem}{ThmFPT}\label{thm:FPT}
	There is an algorithm to check whether a geometric complete bipartite graph with $k$ interior vertices contains a plane Hamiltonian cycle in $O(n \log n + nk^2) + O(k^{5k})$ time.
\end{restatable}
\begin{proof}
	The upper bound on the running time is derived above.
	It remains to show correctness of the algorithm.
	
	First assume that $H$ contains a plane Hamiltonian cycle $C$.
	Let $F$ denote the directed cycle on $I$ obtained by skipping all vertices from $B$ during a counterclockwise traversal of $C$ and let $G_F$ denote the set of edges in $F$ corresponding to skipped paths from~$C$.
	Then $F$ is an initial cycle of $I$ where $G_F$ is the set of gaps.
	As above, replace each gap vertex that is incident to two gaps with a zero-length dummy edge.
	For each gap vertex $u$ in $F$ let $A(u)$ denote the critical arc in $B$ which contains a neighbor of $u$ in~$C$.
	Then $F$ and the selection of critical arcs given by $A$ are Hamiltonian and hence $A$ is feasible by \cref{lem:FPTfeasible}.
	This means, that the algorithm described in \cref{sec:FPT} correctly shows that $H$ has a plane Hamiltonian by identifying the initial cycle $F$ and the feasible selection $A$.
	
	If, the other way round, the algorithm finds an initial cycle of $I$ with a feasible selection of critical arcs, then these are Hamiltonian by \cref{lem:FPTfeasible} and, hence, $H$ contains a plane Hamiltonian cycle.
	\qed
\end{proof}

\section{Completing disjoint bicolored line segments to a plane Hamiltonian cycle is \np-complete}\label{sec:NP}

In 1989, Rappaport~\cite{DBLP:journals/siamcomp/Rappaport89} proved that completing a set of nonintersecting simple polylines to a plane simple Hamiltonian cycle is \np-complete, by reduction from the Hamiltonian path problem on planar cubic graphs~\cite{DBLP:journals/siamcomp/GareyJT76}.
In 2019, Akitaya et al.~\cite{DBLP:conf/compgeom/AkitayaKRST19} and independently in 2021, Jiang, Jiang, and Jiang~\cite{DBLP:journals/corr/abs-2108-12812} showed that this holds for disjoint line segments as well.
We show that both results generalize for a bicolored version of the problem.

\begin{restatable}{theorem}{ThmNPCSegments}\label{thm:npc-segments}
	It is \np-complete to decide if a set of edges of a geometric complete bipartite graph $H$ can be completed to a plane Hamiltonian cycle in $H$.
\end{restatable}

\begin{proof}
We show that the gadgets of Rappaport and Jiang, Jiang, and Jiang still work if they are embedded in a geometric complete bipartite graph.
In particular, we show how to color the vertices in the gadgets such that yes-instances of the uncolored version will become yes-instances of the bipartite graph version (no-instances will remain no-instances in any case).
For our arguments we briefly review the main ideas of the aforementioned reductions.

\begin{figure}
	\centering
	\includegraphics[page=3]{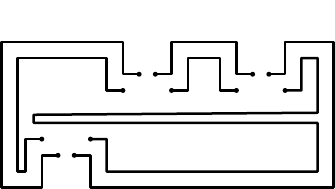}
	\hspace{2ex}
	\includegraphics[page=4]{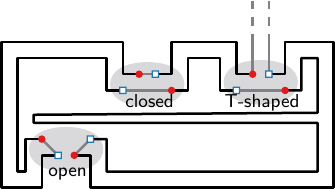}
	\caption[Rappaport's modules]
	{
		Rappaport's modules.
		\subcaptionheading{Left:} A schematic drawing of the module.
		\subcaptionheading{Right:} The three possible completions for a door and the supporting coloring.
	}
	\label{fig:rappaport-modules}
\end{figure}

%In a rectilinear drawing of a graph, vertices are drawn as horizontal line segments, and edges as vertical line segments.
%Though \np-complete for general planar graphs~\cite{DBLP:journals/siamcomp/GargT01}, for triconnected planar graphs such as planar cubic graphs, such a drawing always exists, even on integer coordinates, is planar, and is computable in linear time~\cite{DBLP:journals/dcg/RosenstiehlT86}.
We begin with Rappaport's~\cite{DBLP:journals/siamcomp/Rappaport89} reduction from the Hamiltonian path problem on planar cubic graphs. 
The input graph is drawn such that vertices and edges are represented by noncrossing horizontal and vertical line segments, respectively, such that each vertical line segment connects two horizontal line segments corresponding to adjacent vertices.
Such a drawing exists for every planar cubic graph~\cite{DBLP:journals/dcg/RosenstiehlT86}.
Next, the horizontal segments are replaced by so-called \lightdfn{modules}.
Each module is based on a pair of nested polygons consisting solely of axis-parallel segments, see \cref{fig:rappaport-modules} (left).
Small gaps are cut into both polygons at the three \enquote{attachment points}, where the original horizontal segment connects to vertical segments (recall that the input graph is cubic).
This splits all polygons into polylines.
Each pair of gaps belonging to the same attachment point defines a \lightdfn{door}; see \cref{fig:rappaport-modules} (right).
Finally, all vertical segments are removed.

Rappaport~\cite{DBLP:journals/siamcomp/Rappaport89} proved that the final set of segments can be completed to a plane Hamiltonian cycle if and only if the input graph has a Hamiltonian path.
Also, for any such completion the endpoints at every door can be connected in only three different ways (closed, open, T-shaped); see \cref{fig:rappaport-modules} (right).

In our setting each endpoint of a polyline becomes a vertex of the bipartite host graph $H$ (we will add more vertices to the host later).
We now introduce a coloring for the polyline endpoints, processing each module and each polyline in the module independently.
For polylines of the outer polygon, color the counterclockwise first vertex of the polyline \colonename and the last vertex \coltwoname.
For polylines of the inner polygon, color the counterclockwise first vertex \coltwoname and the last vertex \colonename; see \cref{fig:rappaport-modules} (right).
% leftmost polyline endpoint (if any) on the top line \colonename and the leftmost polyline endpoint (if any) on the bottom line \coltwoname.
%For the polylines coming from the inner polygon, color the leftmost polyline endpoint (if any) on the top line \coltwoname and the leftmost polyline endpoint (if any) on the bottom line \colonename.
%Finally, color all other polyline endpoints of the inner (resp.\ outer) polygon in an alternating fashion \colonename or \coltwoname as indicated in~\cref{fig:rappaport-modules} (right).
We claim that any segment (connecting two polyline endpoints) that might be used by a completion to a plane Hamiltonian cycle is an edge in $H$.
By construction, it is clear that closed and open doors use only edges from $H$.
For each pair of doors from distinct modules one door comes from the upper part of its module and the other from the lower part of its module.
In this case, the colorings of the two doors are inverted to each other, as if they were rotated by $180^\circ$.
So, doors connected via a T-shaped completion use edges from $H$ (joining two modules) as well.
So  we showed that the endpoints of the polylines in resulting yes-instances of Rappaport's reduction can be colored such that any completion to the desired cycle adds only edges from $H$.

\begin{figure}
	\centering
	\includegraphics[page=2]{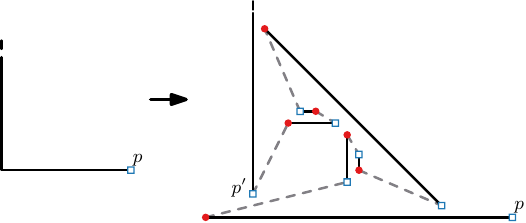}
	\caption[Jiang, Jiang and Jiang's gadget]
	{
		Jiang, Jiang and Jiang's gadget with a suitable coloring.
		In each such coloring, the gadget always passes the color of $p$ on to $p'$.
	}
	\label{fig:jiang-construction}
\end{figure}

Jiang, Jiang, and Jiang~\cite{DBLP:journals/corr/abs-2108-12812} give a reduction from Rappaport's polyline version of the problem to a version that consists only of disjoint line segments.
The reduction replaces each bend of every polyline with a (suitably rotated) gadget consisting of five disjoint line segments; see \cref{fig:jiang-construction}.
Then the authors show that each completion to a noncrossing Hamiltonian cycle adds only segments within the gadgets, or segments already used in the original polyline version.
Again, we reuse the gadget of their reduction.
We add the endpoints of these segments as further vertices to our host graph $H$.
The segments in the gadgets will form the set of edges from our host $H$ that should be completed to a plane Hamiltonian cycle in $H$.
We show that we can color the endpoints of the segments such that yes-instances lead to yes-instances in the bipartite graph version.
%The new line segments have restricted mutual visibility, which imposes a fixed traversal order within any partial non-crossing simple cycle.
%(This is sufficient because all line segments in Rappaport's modules are axis-parallel and thus all corners are right-angle corners.)
We start with a coloring of the polyline endpoints as described above.
The replacements are carried out one by one for every polyline, and the bends in every polyline are processed from one end to the other.
Thus, for each bend that is currently processed, we have one incident segment whose other endpoint $p$ is already colored.
To color the vertices from the gadget, we traverse the segments of the gadget as in \cref{fig:jiang-construction} and color the vertices alternating.
Note that the vertex last colored in the gadget (called $p'$ in \cref{fig:jiang-construction}) receives the same color as $p$.
Hence, the last edge of the polyline(-replacement) is colored correctly, i.e. is from $H$.
This shows that the derived edge set in $H$ can be completed to a plane Hamiltonian cycle in $H$ if and only if the polyline version can be completed to a plane Hamiltonian cycle.
This finishes the proof.
\qed
\end{proof}

\section{Conclusion and open problems}
We have shown that geometric complete multipartite graphs which contain a plane cycle $C$ also admit a shorter plane cycle of at least half the size of $C$.
For some restricted geometric complete bipartite graphs, we have shown that they contain plane cycles of each possible even length.
This prompts the following question (distinguishing between two or more colors).

\begin{problem}
	Let $H$ be a geometric complete multipartite graph that contains a plane cycle of length~$t$. Does $H$ contain a plane cycle of length $t'$ for any (even) $t'$ with $3\leq t'\leq t$?
\end{problem}

%We remark that while \cref{fig:no-c4-chord} shows that it might not be possible to obtain the wanted cycle of length $2(t-1)$ via a chord in~$C$, the point set does admit such a cycle.
   
We have also given an FPT algorithm to decide whether a geometric complete bipartite graph has a plane Hamiltonian cycle. 
It is open whether this problem is NP-hard, which is related to an open question stated by Claverol et al.~\cite{COGST18}.
%Due to our proof that it is NP-complete to decide whether some fixed segments in $P$ can be extended to a plane Hamiltonian cycle and the general NP-completeness of most Hamiltonian cycle questions, we conjecture the following.

\begin{problem}\label{conj:ConjectureNP}
	Is it \np-complete to decide if a given geometric complete multipartite graph contains a plane Hamiltonian cycle?
\end{problem}

\bibliographystyle{splncs04}
\bibliography{lit}

\clearpage
\appendix

\section[Omitted material from Section 2]{Omitted material from \cref{sec:monotonicity}}\label{app:monotonicity}

\begin{lemma}\label{lem:anthropomorphicPolygon}
	Each plane geometric cycle $C$ with exactly three principal points, where the principal points form a 3-vertex path in $C$, has at most five vertices.
\end{lemma}

\begin{proof}
	If $C$ is a triangle, then we are done, so assume $C$ contains at least four vertices.
	A principal point is called an \lightdfn{ear} (a \lightdfn{mouth}) if the edge between its neighbors lies in the interior (exterior) of $C$.
	It is well known that any polygon has at least two ears, and every nonconvex polygon has at least one mouth.
	In a convex polygon, all vertices are ears.
	So $C$ must be nonconvex, and therefore has exactly two ears, call them $x$ and $y$, and one mouth $z$.
	Consider a triangulation $T$ of the interior of $C$ and let $D$ denote the set of edges in $T$ that are not in~$C$.
	Each edge in $D$ splits $C$ into two parts and each part contains one of the ears.
	Moreover, the vertices from $C$ not incident to any edge from $D$ are exactly the ears.
	So the mouth $z$ is incident to some edge from $D$.
	This shows that the mouth $z$ is the middle point in the path formed by the three principal points.
	We also see that $z$ is incident to all edges from $D$, as otherwise there is a third ear or one of the ears ($x$ or $y$) is not a neighbor of $z$.
	Consider the $x$-$y$ path $C'$ in $C$ not containing $z$.
	As $z$ is the only mouth of $C$ and incident to all edges from $D$, we conclude that $C'$ forms a convex chain ($C$ looks like in \cref{fig:anthropomorphicPolygon}).
	We see that at most two of the vertices in $C$ are not ears of $C$ (as $z$ can not \enquote{block} more).
	Hence, $C$ has at most five vertices.
	\qed
\end{proof}

\begin{figure}
	\centering
	\includegraphics{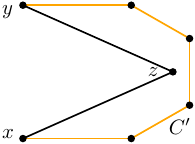}
	\caption{A geometric cycle with two ears $x$ and $y$  and one mouth $z$ forming a path. The chain $C'$ contains further ears if it contains more than four vertices.}
	\label{fig:anthropomorphicPolygon}
\end{figure}

\end{document}